\documentclass[runningheads,a4paper]{llncs}
\usepackage{a4wide}
\usepackage{amsmath,amssymb}
\usepackage{booktabs}
\usepackage{multirow}
\usepackage{graphicx}
\usepackage{hyperref}
\usepackage{multirow}

\usepackage{xcolor}

\usepackage[misc,geometry]{ifsym}

\newcommand{\safhire}{\textsc{Safhire}}
\newcommand{\eps}{\varepsilon}
\newcommand{\R}{\mathbb{R}}
\newcommand{\Z}{\mathbb{Z}}

\newcommand{\cA}{\mathcal{A}}
\newcommand{\sort}{\mathrm{sort}}

\begin{document}

\title{Shuffling is Not Enough:\\
Breaking Permutation-Based Model Confidentiality
in Hybrid FHE Inference}

\titlerunning{Shuffling is Not Enough}
\author{%
Jiseung Kim\inst{1} \and
Hyung Tae Lee\inst{2 \textup{(\Letter)}}
}
\institute{
Jeonbuk National University, 
\email{jiseungkim@jbnu.ac.kr}
\and
Chung-Ang University, 
\email{hyungtaelee@cau.ac.kr}
}
\authorrunning{Kim and Lee}

\maketitle

\begin{abstract}
Hybrid fully homomorphic encryption~(FHE) inference improves the practicality of private inference by letting the server evaluate linear layers homomorphically while the client decrypts and applies nonlinearities.
Recent schemes attempt to protect model confidentiality by returning noisy, output-permuted responses and appealing to shuffle-model differential privacy~(DP).
We show that this protection fails in the correctness regime required by hybrid FHE systems.
For a $d$-input linear layer, $d+1$ admissible queries suffice for exact recovery of a permutation-invariant layer summary, hence for perfect model distinguishability.
We further show that input DP is orthogonal to model confidentiality and that the local-DP premise required for shuffle amplification cannot hold under correctness-bounded noise.
We recover all linear layers of a \safhire{}-style ResNet-20 end-to-end from TFHE transcripts  with zero error, using $d+1$ queries per layer for a total of $5{,}712$ direct queries. 
Under the same query model, we also confirm exact per-layer recovery on pretrained ImageNet-scale CNNs and ViT-B/16.
The leaked spectra enable fingerprinting, lineage attribution, and improved logit-based extraction, while suppressing them destroys inference utility.

\end{abstract}

\section{Introduction}
\label{sec:intro}

Private inference is increasingly expected to protect both sides of the interaction: the client's input should remain hidden from the server, and the server's model should remain hidden from the client. 
Fully homomorphic encryption~(FHE) offers a principled route to input privacy, but evaluating modern neural networks entirely under FHE remains expensive in latency and bandwidth. Hybrid FHE inference sidesteps this bottleneck by partitioning the computation: the server evaluates each linear layer homomorphically, while the client decrypts, applies the nonlinearity in the clear, and re-encrypts for the next layer.
Recent systems such as \safhire{}~\cite{BCD+25} demonstrate this architecture on real convolution and transformer backbones.

However, the hybrid spirit reintroduces a threat that pure FHE had eliminated: because the client observes intermediate layer outputs in the clear, a curious client can in principle read off the server's weight directly.
To close this gap without reverting to the cost of pure FHE, a line of recent work, such as \safhire{}~\cite{BCD+25}, STIP~\cite{YZL24}, Centaur~\cite{LCZ+25}, and PermLLM~\cite{ZCHZ24}, proposes permutation-based model confidentiality. For example, in \safhire{}, the server evaluates a linear layer, applies a secret permutation to the output coordinates, and returns the resulting ciphertext to the client; since ReLU commutes with coordinate permutation, the shuffled representation can be maintained across rounds.
The resulting confidentiality intuition is appealing: FHE decryption noise acts as a local randomizer, the secret output permutation acts as a shuffler, and shuffle-model differential privacy~(DP)~\cite{BBGN19,FMT21} is invoked to argue that only limited model information is exposed.

This paper shows that the intuition is false in the very regime required for correct hybrid FHE inference.
We study a broad class of permutation-based protection schemes in which the client may issue chosen inputs and observe a noisy permuted response to a linear layer on a known quantization lattice.
Our starting point is the structural observation: if the quantization step is $\gamma$, correct decryption requires the additive noise to remain below $\gamma/2$ in magnitude.
Under exactly this bound, response-side randomness can be removed by rounding, and permutation no longer protects the layer beyond relabeling.

We formalize this observation and show that, for a $d$-input linear layer, $d+1$ admissible queries suffice for exact recovery of a permutation-invariant layer summary, hence for perfect model distinguishability.
In the first round of \safhire{}~\cite{BCD+25}, where no hidden input permutation is applied, the same attack recovers the ordered family itself. 
Since \safhire{}'s threat model permits arbitrary chosen inputs without plaintext validation or well-formness proofs, these queries are admissible in practice.
This immediately implies that any two challenge models with different sorted spectrum multisets are perfectly distinguishable in the model confidentiality game.
This result is distribution-agnostic: it depends only on correctness-bounded noise and known quantization, not on the choice of permutation strategy or on any DP parameter attached to the scheme.

Our analysis also clarifies why the shuffle DP justification does not establish model confidentiality. Input DP and model confidentiality quantify different objects: the former compares neighboring inputs under a fixed model, whereas the latter compares different models under fixed queries. More fundamentally, the local-DP premise required for shuffle amplification cannot hold under correctness-bounded noise. Whenever two admissible inputs induce outputs that differ by at least one lattice step in some coordinate, the corresponding per-coordinate supports are disjoint once the noise remains below $\gamma/2$, forcing the local-DP slack parameter to become vacuous. Thus, the privacy amplification argument collapses before it can be applied.

We validate our analysis on TFHE transcripts~\cite{CGGI20} and ImageNet-scale models. 
For a \safhire{}-style ResNet-20, exact recovery holds for all 21 linear layers using $d+1$ queries per layer for a total of $5{,}712$. The same zero-error behavior persists across eight pre-trained ImageNet-scale CNNs and ViT-B/16 under practical precisions and bounded-noise laws. The resulting leakage enables fingerprinting, lineage attribution, and improved logit-based extraction, while no viable privacy/utility operating points remain once the noise is increased enough to suppress the attack.

In summary, our results make three contributions. First, we provide an exact and query-tight recovery theorem for permutation-based protection under correctness-bounded noise. Second, we show that input DP is orthogonal to model confidentiality in this setting and that the local-DP premise required by shuffle amplification is itself unattainable. Finally, we demonstrate concrete security impact and a sharp privacy/utility dilemma on deployed TFHE transcripts and ImageNet-scale vision models.

\subsection{Related Work}

\smallskip\noindent
\textbf{Model extraction.}
Classical extraction attacks~\cite{TZJRR16,OSF19,JCBKP20}
typically require many oracle queries because nonlinearities such as softmax obscure the linear layers.
For example, the cryptanalytic attack of Carlini et al.~\cite{CJM20} uses $2^{18.5}$--$2^{21.5}$ queries on $4$K--$100$K-parameter MLPs.
By contrast, our attack needs only $d{+}1$ queries per layer, totaling $5{,}712$ on a $0.27$M-parameter ResNet-20.
The gap arises because permutation-based protection directly exposes a sorted linear-layer snapshot, a leakage channel already known to be severe in encrypted databases~\cite{NKW15}.
The link between extraction and fingerprinting~\cite{LZK21}, together with membership inference~\cite{SSSS17}, further supports treating identification-level leakage as a model-confidentiality threat.

\smallskip\noindent
\textbf{Impossibility and attacks on ML defenses.}
Conditional impossibility results are known for black-box observational defenses of polynomial-size decision trees under subexponential learning parity with noise~(LPN)~\cite{K23} and for learnability-preserving data encoding~\cite{XSD24}.
Our result is incomparable, giving an \emph{unconditional} guarantee of exact sorted-spectrum recovery on the narrower channel of Definition~\ref{def:perm-scheme}.
ArrowMatch~\cite{WDC+25} empirically shows that TEE-shielded weight obfuscation alone fails against adaptive clients.

\smallskip\noindent
\textbf{Attacks on cryptographic inference stacks.}
Prior attacks are empirical.
SEEK~\cite{CF22} extracts approximate weights from an $11.7$M-parameter
ResNet-18 over homomorphic encryption/multiparty computation~(MPC) using fewer than $50$ queries per parameter,
about $5{\times}10^8$ in total, with relative error about $3{\times}10^{-4}$.
Its leverage is arbitrary intermediate-feature shifts rather than a permuted
linear snapshot, so it does not target
Definition~\ref{def:perm-scheme}.
Works such as~\cite{CCCM23,FT25} attack pure-FHE and split-model FHE LLM
inference outside the permutation class.

\smallskip\noindent
\textbf{Attacks on permutation defenses and DP limits.}
Thomas et al.~\cite{TZC+25} use sequential-token markers to recover client
\emph{input} tokens in PermLLM, STIP, and Centaur, which is disjoint from our
chosen-input attack on \emph{model} parameters.
The limitations of DP as a \emph{model}-protection, rather than training-data,
primitive are discussed in~\cite{DP20,GS22,YLL+22};
Proposition~\ref{prop:premise} quantifies this limitation in the hybrid-FHE
regime.

\section{Background}
\label{sec:background}
We briefly review the \safhire{} setting and the shuffle-model DP argument used by permutation-based protection schemes.
We need only the quantized plaintext lattice, the correctness condition, the per-round permutation interface, and the chosen-input threat model, not the full algebraic machinery of TFHE~\cite{CGGI20}.

\subsection{Hybrid FHE Inference via \safhire{}}
\label{sec:bg-fhe}

We illustrate hybrid FHE inference through
\safhire{}~\cite{BCD+25}, the primary instantiation
of Definition~\ref{def:perm-scheme}.

\smallskip
\noindent\textbf{RLWE encryption (torus framing).}
\safhire{} follows the TFHE torus convention~\cite{CGGI20}:
plaintexts lie on $\mathcal{T}_p := \tfrac{1}{p}\Z/\Z$, where typically
$p = 2^b$ for $b \in \{8,12,16\}$.
Thus the relevant plaintext values lie on a lattice of spacing $\gamma = 1/p$.
A ring learning with error~(RLWE) ciphertext with additive noise $e$ decrypts correctly when
$|e| < 1/(2p)$.
We therefore write $\gamma = 1/p$ for the quantization step and
$\Delta < \gamma/2$ for the correctness-bounded noise magnitude.
This bound is the key operational condition in our analysis: it implies that decoded layer outputs are observed on a known lattice and perturbed by noise smaller than half a lattice step.

\smallskip\noindent\textbf{Batched protocol.}
Inference proceeds in $L$ rounds, one per linear layer. In round $r$, the server first removes the previous round permutation by applying $\sigma_{r-1}^{-1}$ to the incoming activation, homomorphically evaluates $y=W_rx+b_r$, applies the round-$r$ output permutation $\sigma_r$, and returns the resulting ciphertext to the client.
The client decrypts and applies coordinate-wise ReLU before re-encrypting the activation for the next round. 
Because ReLU commutes with coordinate permutation, the shuffled representation can be maintained across rounds.
As in \safhire{}, the final layer is left unshuffled so that the client can read the argmax directly. The permutations are derived from a secret session seed, fixed within a session and re-randomized across sessions~\cite[Section~4.4]{BCD+25}.

\smallskip\noindent\textbf{Threat model.}
We adopt \safhire{}'s threat model~\cite[Section~3]{BCD+25}.
The server is honest-but-curious, whereas the client may send arbitrary inputs and may attempt to reverse-engineer the model.
No distributional restriction, plaintext validation, or zero-knowledge proof of well-formedness is imposed on client queries.
In particular, the input encoding packs a flattened feature chunk into RLWE plaintext coefficients~\cite[Section~4.3]{BCD+25}, so the basis inputs $\{0, e_1, \ldots, e_d\}$ used in our attack are admissible. These properties are precisely the ingredients exploited later: the client observes a shuffled, noisy linear-layer response on a known quantization lattice and may query it adaptively.

\subsection{Shuffle Model of Differential Privacy}
A randomized mechanism $M$ satisfies $(\eps,\delta)$-differential privacy~\cite{DR14}
if $\Pr[M(D)\in S] \le e^\eps \Pr[M(D')\in S] + \delta$
for every neighboring pair $D,D'$ and every measurable set $S$.
In the shuffle model~\cite{BBGN19,FMT21}, $n$ local randomizers first privatize
their inputs, after which a trusted shuffler applies a uniform random permutation
to the resulting messages.
If each local randomizer satisfies $(\eps_0,\delta_0)$-local DP and the parameters
lie in the appropriate regime, the shuffled mechanism enjoys an amplified central-DP
guarantee, with privacy improving roughly as $O(\eps_0/\sqrt{n})$.

Permutation-based defenses in our class instantiate this template by treating FHE
decryption noise as the local randomization step and the secret output-channel
permutation as the shuffler.
The claim is that, because the client observes only a noisy permuted response
rather than the ordered vector $Wx+b$, shuffling amplifies the privacy contributed
by the noise.
Our analysis in Section~\ref{sec:dp} examines this claim in the hybrid-FHE setting.
The central question is whether the correctness-bounded noise required for reliable
decryption can also satisfy the local-DP premise needed for shuffle amplification.

\section{Breaking Permutation-Based Model Protection}
\label{sec:dp}

In this section, we show that permutation-based protection is incompatible with hybrid FHE inference under the correctness-bound noise regime. 
We first formalize the scheme class and prove exact sorted-spectrum recovery from $d+1$ adaptive chosen queries.
We then show that shuffle-DP does not ensure model confidentiality, and that increasing the noise enough to suppress the attack quickly destroys utility.
We begin with notation used throughout the paper.

\smallskip
\noindent\textbf{Notation.}
For $v \in \R^m$, let $\sort(v)$ denote the vector obtained by arranging its entries in non-decreasing order. 
For a linear layer $W \in \R^{m\times d}$ and $b \in \R^m$, define $s_i(W,b) := \sort(W_{:,i}+b)$ as the
\emph{sorted column spectrum} of the $i$-th column $W_{:,i}$, and let $L(W,b) := (s_1(W,b),\ldots,s_d(W,b))$ denote the
resulting ordered family.
By construction, $L(W, b)$ is precisely the information preserved under an output-channel permutation.

\subsection{Permutation-Based Protection Schemes}

\begin{definition}[Permutation-Based Protection]
\label{def:perm-scheme}
Let $\mathcal{X}\subseteq \R^d$ be a plaintext query space containing
$\{0,e_1,\ldots,e_d\}$.
A \emph{permutation-based protection scheme}
for a linear layer $W \in \R^{m\times d}$ and $b \in \R^m$
is a pair of algorithms
$\Pi = (\mathsf{Query}, \allowbreak \mathsf{Respond})$
such that:
\begin{itemize}
    \item $\mathsf{Query}(x)$: on input $x \in \mathcal{X}$, the client produces a well-formed query object $q$.

    \item $\mathsf{Respond}(W,b,q)$: on input $q$, the server determines an effective input
    $x^\star \in \R^d$ and returns a response from which the client obtains $\widetilde{y}=\pi(Wx^\star+b)+\eta,$ where $\pi \in S_m$ is a secret output permutation chosen according to the
    scheme's permutation strategy and $\|\eta\|_\infty<\Delta$.

    If the scheme uses a hidden input permutation, then
$x^\star=\rho(x)$ for some secret $\rho \in S_d$; otherwise $x^\star=x$ (i.e., $\rho= \mathrm{id}$).
\end{itemize}

We say that $\Pi$ has \emph{quantization step} $\gamma$ if, for every
$x \in \{0,e_1,\ldots,e_d\}$ and the corresponding effective input $x^\star$,
the vector $Wx^\star+b$ lies on a lattice of minimum spacing $\gamma$.
\end{definition}

\begin{remark}
The following remarks clarify the operational meaning of Definition~\ref{def:perm-scheme} and its instantiation in FHE-based protocols.
\begin{itemize}
    \item Definition~\ref{def:perm-scheme} subsumes \safhire{} round-$r$ interface by taking
$x^\star = \sigma_{r-1}^{-1}(x)$ and $\pi = \sigma_r$, with $\sigma_0 = \mathrm{id}$
in the first round.
The special case $x^\star = x$ recovers the output-shuffling-only interface.

    \item The client may issue adaptive chosen queries
$x^{(1)},x^{(2)},\ldots \in \mathcal{X}$,
obtain $q^{(t)} \leftarrow \mathsf{Query}(x^{(t)})$,
and observe the corresponding decoded responses
$\widetilde{y}^{(t)}$.

    \item The zero query ensures that $b$ lies on the same lattice as the basis responses.
If the relevant values lie in $\frac{1}{p}\Z$, then one may take $\gamma=1/p$.

    \item In FHE instantiations, $\mathsf{Query}(x)$ is the client-side
encryption of $x$, and $\mathsf{Respond}$ returns a ciphertext whose
decryption yields $\widetilde y$.
The basis queries are admissible whenever arbitrary client activations are
permitted and no validation mechanism rules them out.
\end{itemize}
\end{remark}

Definition~\ref{def:perm-scheme} abstracts a common pattern in which the client observes a permuted, noisy linear-layer response and may issue chosen queries against it.
While motivated here by hybrid FHE inference, the definition itself is not specific to FHE.
\safhire{}'s first layer is the clearest instantiation: the server evaluates a linear layer, applies a secret output permutation, and returns the result under correctness-bounded noise. 
Its later layers additionally apply a server-side input unshuffle, which we analyze separately in Section~\ref{sec:multi-round}.
Section~\ref{sec:discussion} examine Fission~\cite{UMS+25}, STIP~\cite{YZL24}, Centaur~\cite{LCZ+25}, and PermLLM~\cite{ZCHZ24} with respect to the hypotheses of Definition~\ref{def:perm-scheme}.
Among them, Fission is MPC-based rather than FHE-based, and is included only for comparison at
the level of hypotheses rather than as a direct instantiation of the definition.
A scheme that augments client queries with zero-knowledge well-formedness proofs would violate hypothesis~(C) and therefore fall outside the definition; none of the schemes considered here does so.

\begin{theorem}[Exact Sorted-Spectrum Recovery]
\label{thm:impossibility}
Let $\Pi$ be a permutation-based protection scheme under
Definition~\ref{def:perm-scheme}, with noise bound $\Delta$
and quantization step $\gamma \ge 2\Delta$.
Conditioned on $\|\eta\|_\infty < \Delta$, there exists an
efficient adversary $\cA$ that, using $d{+}1$ adaptive
queries, recovers $\sort(b)$ from the zero query and the
sorted-spectrum multiset $\{\!\{\sort(W_{:,j}+b)\}\!\}_{j=1}^d$
from the $d$ basis queries $e_1,\ldots,e_d$.
The recovery is exact on every correctly recovered noise
realization.
\end{theorem}

\begin{proof}
Let $Q_\gamma$ denote rounding to the nearest
$\gamma$-lattice point, and define the round-then-sort operator
$\mathsf{RTS}(\tilde y) := \sort(Q_\gamma(\tilde y)).$

\smallskip\noindent\emph{Bias recovery.}
For the zero query $x=0$, the response satisfies $\tilde y=\pi(b)+\eta.$
Since $b$ lies on the $\gamma$-lattice and $\gamma \ge 2\Delta$, we have
$|\eta_j|<\Delta\le \gamma/2$ for every coordinate~$j$.
Hence rounding recovers $\pi(b)$ exactly, and sorting yields $\sort(b)$.

\smallskip\noindent\emph{Spectrum recovery.}
For each $i\in[d]$, query $x=e_i$.
By Definition~\ref{def:perm-scheme}, the effective input is
$x^\star=\rho(e_i)=e_{\rho(i)}$, so the server computes
$Wx^\star+b = W_{:,\rho(i)}+b,$
and the observed response is $\tilde y=\pi(W_{:,\rho(i)}+b)+\eta.$
The same round-then-sort argument recovers $\sort(W_{:,\rho(i)}+b)$ exactly.
As $i$ ranges over $[d]$ and $\rho\in S_d$ is a bijection,
the adversary recovers the multiset $\{\!\{\sort(W_{:,j}+b): j\in[d]\}\!\}.$
The output is therefore fully determined by $W$ and $b$ on
every noise realization satisfying $\|\eta\|_\infty<\Delta$. \qed
\end{proof}

Theorem~\ref{thm:impossibility} gives the general recovery guarantee under a hidden input permutation~$\rho$.
When $\rho=\mathrm{id}$, as in the first round of \safhire{}, the recovered spectra are additionally labeled by the queried column indices, so the adversary obtains the ordered family $L(W,b)$ rather than merely its underlying multiset.
More generally, the guarantee is distribution-agnostic.
The recovered output is a deterministic function of $W$ and $b$, since all response-side randomness is eliminated by rounding.
Accordingly, the guarantee depends only on the bound $\|\eta\|_\infty<\Delta$ and is therefore independent of the output-permutation strategy, the noise law, and any DP parameter attached to the scheme.
In particular, a DP parameter $\eps$ constrains the ratio of output distributions across neighboring inputs for a fixed model, and thus does not affect any quantity used by the attack.

The hypothesis $\gamma \ge 2\Delta$ in Theorem~\ref{thm:impossibility} is also tight.
Indeed, if $\Delta > \gamma/2$, then no deterministic algorithm can guarantee exact recovery over all noise realizations; see Proposition~\ref{prop:tight} in Appendix~\ref{app:proofs}.

\begin{proposition}[Query Complexity Tightness]
\label{prop:query-tight}
Let $\Pi$ be a permutation-based protection scheme under
Definition~\ref{def:perm-scheme} with $b=0$.
Then any adaptive adversary that recovers
$\{\!\{\sort(W_{:,j})\}\!\}_{j=1}^d$
on every instance requires at least $d$ queries.
Together with Theorem~\ref{thm:impossibility}, this gives exact minimax
query complexity: $d$ for bias-free layers, and $d+1$ when $\sort(b)$ must
also be recovered.
\end{proposition}

\begin{proof}
Fix any adaptive strategy using $k \le d-1$ queries
$x^{(1)},\ldots,x^{(k)} \in \R^d$.
The subspace $V = \mathrm{span}\{x^{(1)},\ldots,x^{(k)}\}$
has dimension at most $k \le d-1 < d$,
so there exists a unit vector $v \in \R^d$ with
$v \perp x^{(t)}$ for all $t \in [k]$.
Define $W' := W + uv^T$ for any nonzero $u \in \R^m$.
For every query $x^{(t)}$,
$W'x^{(t)} = Wx^{(t)} + u\underbrace{(v^Tx^{(t)})}_{=0} = Wx^{(t)},$
so $(W',0)$ and $(W,0)$ produce identical responses
under any permutation strategy and noise realization.
The adversary's view is identical, hence it outputs
the same answer for both models.
Since $v \ne 0$, there exists $j^* \in [d]$ with
$v_{j^*} \ne 0$, so $W'_{:,j^*} = W_{:,j^*} + uv_{j^*}$.
For generic $W$ and $u$, this column is distinct from
every column of $W$, so the sorted-spectrum multisets differ.
The adversary therefore errs on at least one of
$(W,0)$ and $(W',0)$. \qed
\end{proof}

\begin{corollary}
\label{cor:safhire}
The first round of \safhire{}, which is directly queried by the client, satisfies the hypotheses of Theorem~\ref{thm:impossibility}.
At precision $p$, it has $\gamma=1/p$ and $\Delta=1/(2p)$, so $\gamma=2\Delta$, and the condition $\|\eta\|_\infty<\gamma/2$ is exactly the FHE correctness condition.
Under standard lattice-based FHE parameters, this condition fails only with negligible probability, typically below $2^{-128}$~\cite{BCD+25}.
Hence, over at most $\mathrm{poly}(\lambda)$ attack queries, exact sorted-spectrum recovery holds with probability $1-\mathsf{negl}(\lambda)$.
\end{corollary}

\subsection{Input DP and Model Confidentiality Are Orthogonal}

We begin by emphasizing that input DP and model
confidentiality are distinct notions.
For a fixed model~$W$, input DP bounds the ratio of
output distributions induced by two neighboring
inputs $x$ and $x'$.
Model confidentiality, by contrast, concerns the
distinguishability of two different models under a
fixed input.
The two quantifications are therefore logically
independent.
Nevertheless, prior permutation-based proposals
invoke input DP for model
confidentiality.

\begin{definition}
\label{def:mc-game}
The adversary $\cA$ selects challenge models
$(W_0,b_0)$ and $(W_1,b_1)$ whose sorted-spectrum
multisets differ, i.e., 
$\{\!\{s_i(W_0,b_0): i \in [d]\}\!\}
\ne
\{\!\{s_i(W_1,b_1): i \in [d]\}\!\},$
and submits them to the challenger.
The challenger samples
$\beta \xleftarrow{\$} \{0,1\}$ and grants $\cA$
oracle access to
$\Pi.\mathsf{Respond}(W_\beta,b_\beta,\cdot)$.
$\cA$ then issues adaptive queries and
outputs a bit $\hat{\beta}$.
We define the model confidentiality advantage of $\cA$ by
$\mathsf{Adv}^{\mathrm{mc}}_{\Pi,\cA}
:= |2\Pr[\hat{\beta}=\beta]-1|.$
\end{definition}

\begin{proposition}
\label{prop:dp-direction}
Let $\Pi$ be a permutation-based protection scheme under Definition~\ref{def:perm-scheme}
with $\gamma \ge 2\Delta$.
Then, there exists an adversary $\cA$ issuing $d$ chosen queries such that $\mathsf{Adv}^{\mathrm{mc}}_{\Pi,\cA} = 1$, regardless of any input DP guarantee satisfied by  
$\Pi.\mathsf{Respond}(\cdot\,; W)$.
\end{proposition}

\begin{proof}
The adversary $\cA$ issues the $d$ basis queries
$e_1,\ldots,e_d$.
By Theorem~\ref{thm:impossibility}, the round-then-sort
procedure recovers
$\sort(W_{\beta,:,\,\rho(i)}+b_\beta)$ for each
$i \in [d]$, and hence the $d$ responses reveal the
multiset $\mathcal{M}_\beta :=
\{\!\{\sort(W_{\beta,:,j}+b_\beta)\}\!\}_{j=1}^d.$
Since $\rho$ is a bijection on $[d]$,
this multiset is recovered regardless of the
hidden permutation $\rho$.
By Definition~\ref{def:mc-game},
$\mathcal{M}_0 \ne \mathcal{M}_1$ as multisets, so $\cA$ outputs
the unique bit $\beta$ whose multiset matches the
observation and 
$\mathsf{Adv}^{\mathrm{mc}}_{\Pi,\cA}=1$.
By contrast, input DP constrains single-model ratios
of the form
$\Pr[f(x;W)\in S]/\Pr[f(x';W)\in S]$, which are
orthogonal to the two-model comparison required in
Definition~\ref{def:mc-game}. \qed
\end{proof}

\subsection{The Shuffle DP Premise Cannot Hold}
\label{sec:shuffle-vacuous}

Shuffle amplification is conditional on an
$(\eps_0,\delta_0)$-local DP input. Under the FHE
correctness bound, each local mechanism moves its
output by less than~$\gamma/2$, so neighboring inputs
whose outputs differ by one lattice step produce
disjoint per-coordinate supports. The non-overlapping
region captures the \emph{entire} support of one side,
driving $\delta_0$ to~$1$ and voiding any downstream
amplification.

\begin{proposition}
\label{prop:premise}
Let $M_k(x) = (Wx+b)_k + \eta_k$ denote the per-coordinate
output of a permutation-based scheme in
Definition~\ref{def:perm-scheme}, and assume that each
noise coefficient $\eta_k$ lies almost surely in the
interval $(-B, B)$, i.e., in the correctness-conditioned
bounded-noise regime with $B < \gamma/2$.
Let $x, x'$ be inputs such that, for some coordinate $k$,
the sensitivity satisfies
$\Delta f := \left|(Wx+b)_k - (Wx'+b)_k\right| \ge \gamma,$
where $\Delta f$ denotes DP sensitivity and is unrelated
to the noise bound~$\Delta$.
Then any $(\eps_0,\delta_0)$-DP guarantee for $M_k$ on
the pair $(x,x')$ with finite~$\eps_0$ necessarily
satisfies $\delta_0 \ge 1$.
Equivalently, $M_k$ cannot satisfy $(\eps_0,\delta_0)$-DP
for any $\delta_0 < 1$.
\end{proposition}

\begin{proof}
Let $I_x := \mathrm{supp}(M_k(x))$ and
$I_{x'} := \mathrm{supp}(M_k(x'))$.
Write
$c_x := (Wx+b)_k$ and $c_{x'} := (Wx'+b)_k$.
Since $\eta_k \in (-B,B)$ almost surely, we have
$I_x \subseteq [c_x - B,\, c_x + B]$ and 
$I_{x'} \subseteq [c_{x'} - B,\, c_{x'} + B]$,
where we emphasize that these are containments, not necessarily equalities.
Now consider the set $S := I_x \setminus I_{x'}.$
By construction, $S$ contains exactly those outputs that are attainable under
$x$ but impossible under $x'$.
Applying the $(\eps_0,\delta_0)$-DP condition to the event $S$ yields
$\Pr[M_k(x) \in S]
\le
e^{\eps_0}\Pr[M_k(x') \in S] + \delta_0.$
Since $S \cap I_{x'} = \emptyset$, we have
$\Pr[M_k(x') \in S] = 0$, and therefore
$\Pr[M_k(x) \in S] \le \delta_0.$
Under the bound $B < \gamma/2$, any coordinate gap
$\Delta f \ge \gamma$ implies $\Delta f > 2B$.
Hence the closed intervals
$[c_x - B,\, c_x + B]$ and $[c_{x'} - B,\, c_{x'} + B]$
are disjoint and separated by a positive gap.
It follows that $I_x \cap I_{x'} = \emptyset$, in fact
$S = I_x$.
So, $\Pr[M_k(x) \in S] = \Pr[M_k(x) \in I_x] = 1.$
Combining with the previous inequality gives
$1 \le \delta_0.$ \qed
\end{proof}

The above proof uses only the almost-sure support bound $\eta_k \in (-B,B)$.
In particular, it does not require a density or absolute continuity, and
therefore applies equally to discretely supported RLWE noise and to
almost-surely bounded continuous surrogates.

\begin{remark}[Beyond $(\eps_0,\delta_0)$-DP]
\label{rem:dp-variants}
Proposition~\ref{prop:premise} assumes almost-surely bounded noise
$\eta_k \in (-B,B)$ with $B < \gamma/2$, i.e., the correctness-conditioned
regime. This covers practical discrete-Gaussian RLWE noise as well as clipped
continuous surrogates. By contrast, an unbounded Gaussian falls outside the
hypothesis, since then $I_x=\mathbb{R}$ and the disjoint-support argument
fails; moreover, any Gaussian large enough to overlap adjacent lattice points
already violates $|\eta| < \gamma/2$ with non-negligible probability.
Within the bounded-noise regime, disjoint supports imply total variation
$\mathrm{TV}(M_k(x),M_k(x'))=1$, so any divergence-based privacy notion that
upper-bounds total variation, including standard relaxations of
$(\eps_0,\delta_0)$-DP, becomes infinite or vacuous on such neighbors.
The same obstruction extends to group privacy, since a single lattice-step gap
$\Delta f \ge \gamma$ already gives $\Delta f > 2B$.
\end{remark}

In the \safhire{} setting, a per-coordinate gap of~$\gamma$ is unavoidable.
For the basis queries used in Theorem~\ref{thm:impossibility}, moving from
$x=0$ to $x=e_i$ changes the $k$-th output coordinate by the entry of
$W_{:,i}$, which lies on the quantization lattice of spacing~$\gamma$.
Hence, any scheme with at least one nonzero weight entry admits neighboring
inputs to which Proposition~\ref{prop:premise} applies.
Since $\delta_0 \ge 1$ renders the local-DP premise of shuffle amplification
trivially impossible to satisfy, no choice of shuffle size~$n$, target~$\delta$,
or Gaussian-surrogate noise can restore the guarantee.

One could try to restrict the shuffle-amplification analysis to neighbors with
$\Delta f < \gamma$, thereby placing them outside the scope of
Proposition~\ref{prop:premise}.
In \safhire{}, however, this restriction is unavailable:
its threat model~\cite[Section~3]{BCD+25} allows the client to issue arbitrary
inputs, with no bound on $\|x\|$, so the basis queries of
Theorem~\ref{thm:impossibility} are admissible and the vacuity result applies
directly.
Even a modified variant that artificially restricted sensitivity would recover
only an input-privacy guarantee, which remains orthogonal to model
confidentiality by Proposition~\ref{prop:dp-direction}.

\subsection{Noise/Utility Tradeoff}
\label{sec:tradeoff-thm}
Theorem~\ref{thm:impossibility} rules out confidentiality under the correctness bound $B < \gamma/2$. 
A natural remedy is to relax this bound and inject larger noise, trading decryption accuracy for defense.
We establish a distribution-free lower bound on the noise required to defeat a Lipschitz-fingerprint attack, instantiated in Section~\ref{sec:impact}.

\begin{proposition}
\label{prop:relaxed-recovery}
Let $\Pi$ be a permutation-based protection scheme under
Definition~\ref{def:perm-scheme} with quantization step~$\gamma$,
and let $Q_\gamma$ denote rounding to the nearest $\gamma$-lattice point.
Assume a relaxed noise bound $B > 0$, not necessarily with $B \le \gamma/2$,
and define
$\tilde{s}_i := \sort(Q_\gamma(\pi(W_{:,i}+b)+\eta))$.
Then, for every column~$i$ and every noise realization with
$\|\eta\|_\infty < B$,
\[
  \|\tilde{s}_i - s_i(W,b)\|_\infty
  \le e_B
  := \gamma \Big\lfloor \tfrac{B}{\gamma} + \tfrac{1}{2} \Big\rfloor
  \le B + \tfrac{\gamma}{2}.
\]
\end{proposition}

\begin{proof}
Fix a column $i$ and write $v := \pi(W_{:,i}+b) \in \gamma\Z^m.$
Since $\|\eta\|_\infty < B$, each coordinate satisfies $|\eta_j| < B$, and thus $|Q_\gamma(v_j+\eta_j)-v_j|
\le
\gamma \Big\lfloor \tfrac{B}{\gamma}+\tfrac12 \Big\rfloor
= e_B.$
Therefore, $\|Q_\gamma(v+\eta)-v\|_\infty \le e_B.$
Since $\sort$ is invariant under the permutation~$\pi$, and order statistics are non-expansive in $\ell_\infty$,
\[
\|\tilde{s}_i-s_i(W,b)\|_\infty
=
\|\sort(Q_\gamma(v+\eta))-\sort(v)\|_\infty
\le
\|Q_\gamma(v+\eta)-v\|_\infty
\le e_B.
\]
Moreover, since $\lfloor \tfrac{B}{\gamma}+\tfrac{1}{2} \rfloor \le \tfrac{B}{\gamma}+\tfrac{1}{2}$, we have
$e_B
\le
\gamma \Big( \tfrac{B}{\gamma}+\tfrac12 \Big)
=
B+\tfrac{\gamma}{2}.$ \qed
\end{proof}

Parameterizing the bound by the Lipschitz constant of the
fingerprint map makes it directly applicable to the attacks
in Section~\ref{sec:impact}.

\begin{theorem}[Noise/Utility Tradeoff]
\label{thm:tradeoff}
Let $\Pi$ be a permutation-based protection scheme under Definition~\ref{def:perm-scheme}
with relaxed noise bound $B > 0$, let $\mathcal{F}$ be a finite
model family, and let $\phi \colon L \mapsto \phi(L) \in (\R^k,\|\cdot\|_q)$
be a fingerprint map satisfying
\begin{equation}
  \label{eq:fp-lipschitz}
  \|\phi(L) - \phi(L')\|_q
  \le
  L_\phi \cdot \max_i \|s_i - s'_i\|_\infty
  \qquad \text{for some } L_\phi > 0.
\end{equation}
Define $\Delta_\mathcal{F}^\phi
:=
\min_{M \ne M' \in \mathcal{F}}
\|\phi(L(M)) - \phi(L(M'))\|_q$,
and let $\cA_\phi$ denote the round-then-sort nearest-neighbor
classifier under $\phi$ given the $d+1$ chosen queries
$x = 0, e_1, \ldots, e_d$.
Then the followings hold, independently of the noise distribution:
(i) if $L_\phi e_B < \Delta_\mathcal{F}^\phi/2$, then $\cA_\phi$
correctly identifies every $M \in \mathcal{F}$ on every noise
realization satisfying $\|\eta\|_\infty < B$.
(ii) Conversely, defeating $\cA_\phi$ on any $M \in \mathcal{F}$
requires $B \ge \max\{0,\,\Delta_\mathcal{F}^\phi/(2 L_\phi) - \gamma/2\}.$
\end{theorem}

\begin{proof}
Fix the true model $M = (W,b) \in \mathcal{F}$ with sorted spectra $(s_i)$,
and let $(\tilde{s}_i)$ be the recovered vectors.
By Proposition~\ref{prop:relaxed-recovery},
$\max_i \|\tilde{s}_i - s_i\|_\infty \le e_B.$
Write $\tilde{L} := (\tilde{s}_i)$ and $L := L(M)$.
Applying \eqref{eq:fp-lipschitz} to $\tilde{L}$ and $L$ gives $\|\phi(\tilde{L}) - \phi(L)\|_q \le L_\phi e_B.$

\emph{(i)}
Assume $L_\phi e_B < \Delta_\mathcal{F}^\phi/2$.
For any $M' \ne M$ in $\mathcal{F}$, the reverse triangle inequality and
the definition of $\Delta_\mathcal{F}^\phi$ yield
\begin{align*}
\|\phi(\tilde{L}) - \phi(L(M'))\|_q
&\ge
\|\phi(L) - \phi(L(M'))\|_q
-
\|\phi(\tilde{L}) - \phi(L)\|_q \\
&\ge
\Delta_\mathcal{F}^\phi - L_\phi e_B
>
L_\phi e_B
\ge
\|\phi(\tilde{L}) - \phi(L)\|_q.
\end{align*}
Hence $\cA_\phi$ uniquely returns $M$.
\emph{(ii)}
This is immediate by contraposition.
If $\cA_\phi$ fails on some $M \in \mathcal{F}$, then
$L_\phi e_B \ge \Delta_\mathcal{F}^\phi/2$, and therefore
$e_B \ge \frac{\Delta_\mathcal{F}^\phi}{2L_\phi}.$
By Proposition~\ref{prop:relaxed-recovery},
$e_B \le B + \frac{\gamma}{2}.$
Combining the two inequalities gives
$B + \frac{\gamma}{2} \ge e_B \ge \frac{\Delta_\mathcal{F}^\phi}{2L_\phi},$
hence
$B \ge \frac{\Delta_\mathcal{F}^\phi}{2L_\phi} - \frac{\gamma}{2}.$
Taking the maximum with $0$ removes the vacuous case in which the
right-hand side is negative. \qed
\end{proof}

\begin{remark}
\label{rem:colnorm-specialization}
Section~\ref{sec:impact} uses the column-norm fingerprint
$\phi(L) := \sort(\|s_1\|_2,\ldots,\allowbreak \|s_d\|_2)$
with $\ell_2$ nearest-neighbor.
This map satisfies Eq.~\eqref{eq:fp-lipschitz} with
$L_\phi = \sqrt{md}$.
Indeed, if $\max_i \|\tilde{s}_i - s_i\|_\infty \le e$, then
the reverse triangle inequality and
$\|\cdot\|_2 \le \sqrt{m}\|\cdot\|_\infty$ give
$|\|\tilde{s}_i\|_2 - \|s_i\|_2| \le \sqrt{m}\,e$
for each $i$.
Since $\sort$ is non-expansive in $\ell_\infty$ and
$\|\cdot\|_2 \le \sqrt{d}\|\cdot\|_\infty$ on $\R^d$, it follows that
$\|\phi(\tilde L)-\phi(L)\|_2 \le \sqrt{md}\,e$.
Thus Eq.~\eqref{eq:fp-lipschitz} holds with $L_\phi=\sqrt{md}$.
For ResNet-20 \texttt{conv1}, where $m=16$ and $d=27$,
at 8-bit precision with $\Delta_\mathcal{F}^\phi \approx 0.298$,
Theorem~\ref{thm:tradeoff} yields
$B \ge 2.7(\gamma/2)$.
This bound is loose and remains well below the empirical
transition in Table~\ref{tab:tradeoff}.
\end{remark}

Theorem~\ref{thm:tradeoff} certifies attack success below the threshold,
not defender safety above it.
The certified bound $2.7(\gamma/2)$ is two to three orders of magnitude below
the empirical transition at $10^2$--$10^3(\gamma/2)$ in
Table~\ref{tab:tradeoff}, because the worst-case bound
$L_\phi=\sqrt{md}$ assumes adversarial alignment across all $m$ coordinates,
whereas symmetric noise gives a tighter $\sqrt{m}$ scaling.
Together, Theorem~\ref{thm:impossibility} covers $B<\gamma/2$,
$\delta_0 \ge 1$ rules out the shuffle route, and
Theorem~\ref{thm:tradeoff} ties brute-force noise to intrinsic fingerprint
separation.

\vspace{-3mm}
\section{Chosen-Input Attack}
\label{sec:attack}

In this section, we instantiate the recovery attack against
\safhire{}~\cite{BCD+25}, the primary hybrid-FHE
instantiation of Definition~\ref{def:perm-scheme}.

\subsection{Multi-Round Layers}
\label{sec:multi-round}

In \safhire{}'s multi-round protocol, the server applies
$\sigma_{r-1}^{-1}$ before evaluating layer~$r$ and
$\sigma_r$ afterward. Accordingly, a basis query $e_i$
issued in round~$r$ probes column $\sigma_{r-1}^{-1}(i)$
of $W_r$. Since $\sigma_{r-1}$ remains fixed within a
session, the $d$ basis queries in round~$r$ collectively
probe every column of $W_r$ exactly once. Moreover,
\safhire{}'s threat model explicitly permits chosen-input
queries~\cite[Section~3]{BCD+25}, and the basis vectors
are admissible plaintext polynomials under its coefficient
encoding~\cite[Section~4.3]{BCD+25}.

\begin{proposition}
\label{prop:fresh-perm}
Under the hypotheses of
Theorem~\ref{thm:impossibility}, if $\sigma_{r-1}$
is sampled independently for each inference request,
then all distinct sorted spectra of $W_r$ can be
recovered using $O(d\log d)$ queries with probability
at least $1-e^{-c}$.
More precisely, recovery succeeds if the
number $T$ of queries satisfies
$T \ge d(\ln d^* + c),$
where $d^* \le d$ denotes the number of distinct
spectra.
\end{proposition}

\begin{proof}
Each query probes a uniformly random column.
If a distinct spectrum~$s$ appears with multiplicity $m_s$, then a single query
recovers $s$ with probability $m_s/d \ge 1/d$.
By Theorem~\ref{thm:impossibility}, the queried spectrum is recovered exactly
from each response, and duplicate recoveries can be removed by equality testing
on the sorted vectors.
Fix a distinct spectrum $s$.
The probability that $T$ independent queries all miss $s$ is at most $(1 - 1/d)^T \le e^{-T/d}$.
Applying a union bound over the $d^*$ distinct spectra, the probability that
some spectrum is missed is at most $d^* e^{-T/d}.$
Thus, whenever $T \ge d(\ln d^* + c)$, the failure probability is at most
$d^* e^{-T/d} \le d^* e^{-(\ln d^*+c)} = e^{-c}.$
Therefore, recovery succeeds with probability at least $1-e^{-c}$. \qed
\end{proof}

In the within-session regime
(described in Section~\ref{sec:bg-fhe}), the $d$ basis queries recover
$L(W_r,b_r)$ within the overall $d{+}1$ query budget of
Theorem~\ref{thm:impossibility}. 
By contrast, when the permutation is resampled independently for each query, the query complexity becomes $\Theta(d\log d)$.
On trained ResNet-20, we observe $d^*=d$ in every layer, i.e., no duplicate spectra arise.
Consequently, the bound $T \ge d(\ln d + c)$ yields approximately 
$5{,}400$ queries for the widest layer ($d{=}576$, $c{=}3$).
In simulation, the independently-resampled-permutation attack converges with $T_{\mathrm{emp}}/T_{\mathrm{theory}} \in
[0.58,\,0.87]$ over 5 trials per layer.

\medskip\noindent\textbf{What is recovered?}
For each layer the adversary recovers $L(W,b)$ and hence
every permutation-invariant statistic
derived from it, including the affine column norms
$\|W_{:,i}+b\|_2$ and the empirical value
distribution of each column. For bias-free layers, where $b = 0$ as is standard under batch normalization,
this reduces to the
exact weight multisets $\sort(W_{:,i})$, their column
norms, and the Frobenius norm~$\|W\|_F$.
Since $\sort$ does not
commute with subtraction, the bias cannot be disentangled
from the weights in general; however, this residual
ambiguity does not affect any of the applications in Section~\ref{sec:impact}.

\safhire{}~\cite[Section~5, scenario~(iii)]{BCD+25}
treats leakage of only the weight histogram as benign
residual leakage, purportedly bounded by a shuffle-model DP
guarantee. Section~\ref{sec:dp} shows that this guarantee is vacuous: Proposition~\ref{prop:dp-direction} establishes that the DP argument is formulated in the wrong direction, while Proposition~\ref{prop:premise} shows that its local-DP premise fails under the correctness bound.
Section~\ref{sec:impact} further demonstrates that the
recovered spectra suffice for reliable fingerprinting and
lineage detection, and that suppressing these attacks requires noise at scales for which inference utility itself collapses (see Table~\ref{tab:tradeoff} in Section~\ref{sec:impact}).

\subsection{Experimental Validation}
We validate the attack at two levels of concreteness:
a deployment-faithful TFHE transcript replay on \safhire{}'s backend in Table~\ref{tab:tfhe_transcript}, and simulated per-layer queries on eight ImageNet-pretrained architectures in Table~\ref{tab:imagenet_recovery}.
Our CIFAR-10 evaluation uses ResNet-20 with 21 bias-free convolutional layers after BN folding, 0.27M parameters, mean test accuracy 82.8\%, and a per-model query budget of $\sum_{l=1}^{21}(d_l+1)=5{,}712$.
We also confirm zero-error recovery for randomly initialized models, six bounded
noise distributions, and MLPs with nonzero bias; see
Appendix~\ref{app:experiments}.
The codes for all experiments in this paper are available at \url{https://github.com/JiseungKim90/permutation-confidentiality-experiments}.

\smallskip
\noindent\textbf{TFHE transcript replay.}
To rule out simulation artifacts, we replay the attack end-to-end through
Concrete~v2.11~\cite{ConcretePython}, the TFHE backend of
\safhire{}~\cite{CGGI20,BCD+25}.
Table~\ref{tab:tfhe_transcript} reports the results across all 21 convolutional layers
under both a fixed permutation and a fresh-per-query permutation.

\begin{table}[t]
\centering
\caption{TFHE transcript on all 21 convolutional layers of a trained ResNet-20
($p{=}256$, $11{,}424$ encrypted queries total, $5{,}712$ per regime).
Max recovery error $=0$ for every layer in both fixed-perm and
fresh-perm regimes.}
\label{tab:tfhe_transcript}
\setlength{\tabcolsep}{4pt}
\footnotesize
\begin{tabular}{@{}lrrr|lrrr@{}}
\toprule
Shape & \#layers & $d{+}1$ & Total $q$ & Shape & \#layers & $d{+}1$ & Total $q$ \\
\midrule
$16{\times}27$  & 1 & 28  &    28 &
$32{\times}288$ & 5 & 289 & 1{,}445 \\
$16{\times}144$ & 6 & 145 &   870 &
$64{\times}32$  & 1 &  33 &    33 \\
$32{\times}16$  & 1 &  17 &    17 &
$64{\times}288$ & 1 & 289 &   289 \\
$32{\times}144$ & 1 & 145 &   145 &
$64{\times}576$ & 5 & 577 & 2{,}885 \\
\midrule
\multicolumn{3}{@{}l}{Total (21 layers)} & 5{,}712 &
\multicolumn{3}{l}{Max error (all)} & \textbf{0} \\
\bottomrule
\end{tabular}
\vspace{-5mm}
\end{table}

\smallskip
\noindent\textbf{Architecture-agnostic ImageNet validation.}
We apply the simulated attack to eight torchvision-pretrained ImageNet models, namely ResNet-50, ResNet-101, ResNet-152, VGG-16, DenseNet-121, MobileNet-V2, EfficientNet-B0, and ConvNeXt-Tiny.
These models comprise 600 convolutional layers in total, with up to 58M parameters and a largest layer of size $2048\times 4608$.
Every layer exhibits zero recovery error at all three precision levels and under all six bounded noise distributions, as reported in Table~\ref{tab:imagenet_recovery}.
Applying the same pipeline to the 37 linear layers of torchvision
ViT-B/16 with \texttt{IMAGENET1K\_V1} weights, which
contains 86M parameters, likewise yields $\mathrm{max\_err}=0$ at $p=256$, thereby extending the recovery to transformer architectures.
These 37 layers are exactly the attention query/key/value and output projections and the MLP layers used in transformer blocks. Because the same linear primitives underlie language transformers, including the model targeted by PermLLM~\cite{ZCHZ24}, our recovery applies without modification.

\vspace{-5mm}
\begin{table}[h!]
\centering
\caption{Exact sorted-spectrum recovery on eight ImageNet-pretrained architectures in per-layer simulation.
Max recovery error=0 for every convolutional layer at 4-, 8-, and 12-bit precision across all architectures.
The largest layer has shape $2048 \times 4608$.}
\label{tab:imagenet_recovery}
\footnotesize
\setlength{\tabcolsep}{5pt}
\begin{tabular}{@{}lrr|lrr@{}}
\toprule
Architecture & Layers & Params & Architecture & Layers & Params \\
\midrule
ResNet-50    & 53  & 23{,}454{,}912 & ResNet-101      & 104 & 42{,}382{,}016 \\
ResNet-152   & 155 & 57{,}992{,}384 & VGG-16          &  13 & 14{,}710{,}464 \\
DenseNet-121 & 120 &  6{,}870{,}208 & MobileNet-V2    &  52 &  2{,}189{,}760 \\
EfficientNet-B0 & 81 & 3{,}956{,}192 & ConvNeXt-Tiny  &  22 &  1{,}877{,}472 \\
\bottomrule
\end{tabular}
\vspace{-10mm}
\end{table}

\section{Practical Impact of Sorted Spectrum Leakage}
\label{sec:impact}

The defender's residual claim in Section~\ref{sec:attack} is that the loss of column labels renders the recovered multisets benign: even if $L(W,b)$ is exposed, the adversary still cannot
reconstruct the weight matrix, even up to row and
column permutations. This characterization,
however, substantially understates the leakage.
Many security-relevant quantities are
permutation-invariant and hence are fully determined
by $L(W,b)$. These include column norms, the
per-column empirical distributions $s_i(W,b)$, and
the weight-value histogram. This leakage has three
concrete consequences: \emph{fingerprinting},
\emph{lineage detection} for fine-tuned descendants,
and accelerated logit-based \emph{extraction}.
Moreover, suppressing any of these attacks requires
noise at a scale that materially degrades inference
utility, as shown in
Section~\ref{sec:tradeoff}.

\subsection{Attribution: Fingerprinting and Lineage}
\label{sec:attribution}
Even for a fixed architecture, independently trained
models converge to different local minima under
different random seeds. Their sorted column spectra
are therefore deterministically distinct and can serve
as reliable model identifiers.
Using the 20 trained ResNet-20 checkpoints from
Section~\ref{sec:attack}, we define the
\emph{fingerprint} of a model as the sorted
column-norm vector $\phi(W,b) := \sort(\|s_1\|_2,\ldots,\|s_d\|_2)$
of the \texttt{conv1} layer, which is the specialization
described in
Remark~\ref{rem:colnorm-specialization}, and compare
models by $\ell_2$ distance in the fingerprint space.
Since $\phi$ is invariant under column permutations,
the multiset-level recovery guaranteed by
Theorem~\ref{thm:impossibility}, which is the general
case under a secret $\rho \ne \mathrm{id}$, already
suffices for both fingerprinting and lineage
attribution; the stronger labeled recovery available
when $\rho=\mathrm{id}$ is unnecessary.
Since recovery is exact under the
correctness bound, the within-model distance is zero across all noise
realizations, whereas the minimum between-model distance is
$\Delta_\mathcal{F}^\phi = 0.298$.
Consequently, all $\binom{20}{2}=190$ model pairs are
perfectly separable under $\phi$.
On CIFAR-100, the same fingerprint yields
$\Delta_\mathcal{F}^\phi = 0.565$ after 200 training
epochs, and all 190 pairs remain perfectly
separable.
Each such pair also instantiates the model
confidentiality game of
Definition~\ref{def:mc-game} empirically.
In particular, the round-then-sort adversary of
Proposition~\ref{prop:dp-direction} achieves
$\mathsf{Adv}^{\mathrm{mc}}=1$ on every one of the
190 challenge pairs using the $d$ basis queries of
Theorem~\ref{thm:impossibility}, or a single query in
the first-round setting where $\rho=\mathrm{id}$.

\smallskip
\noindent\textbf{Lineage detection.}
Fine-tuning typically perturbs the weights incrementally, so the sorted column spectrum of a
fine-tuned child remains closer to that
of its parent than to that of an independently
trained checkpoint.
To quantify this effect, we select 5 of the 20
ResNet-20 checkpoints as base models, fine-tune three
children from each base under mild, moderate, and
aggressive regimes, and classify each child against
all 20 bases using $\ell_2$ nearest-neighbor search
on the sorted column norms concatenated across three
convolutional layers.
The resulting separation ratios, defined as the
minimum between-base distance divided by the maximum
within-base distance, are 26.7, 6.0, and 1.59. All
exceed 1, yielding 15/15 correct attributions and
zero false positives over 285 non-lineage pairs.

At ImageNet scale, the same fingerprint map applied
to torchvision-pretrained ResNet-50 V1 and V2, along
with a lineage variant $V_1^{\mathrm{ft}}$ obtained by
fine-tuning V1 for two epochs on Imagenette, gives
$\|\phi(V_1)-\phi(V_2)\|_2 = 1106.43$ and
$\|\phi(V_1)-\phi(V_1^{\mathrm{ft}})\|_2 = 1.38$.
Since $V_1^{\mathrm{ft}}$ remains within the
fine-tuning neighborhood of $V_1$,
$\|\phi(V_2)-\phi(V_1^{\mathrm{ft}})\|_2$
agrees with $\|\phi(V_1)-\phi(V_2)\|_2$ to four
significant digits. This yields an approximate
$800{:}1$ lineage-to-non-lineage distance ratio on a
25.5-million-parameter backbone.

\vspace{-3mm}

\subsection{Distillation Ablation}
\label{sec:distillation}

The leaked column norms further facilitate extraction: adding the
recovered spectral prior to standard logit distillation raises
student accuracy by up to $3.1$~percentage points over the logit-only baseline,
with the permutation-invariant \emph{Col-norms} statistic capturing
most of the gain at large query budgets.
Full results across query budgets $1{,}000$--$5{,}000$ and two
architectures are in Appendix~\ref{app:distillation}.

\subsection{The Privacy/Utility Dilemma}
\label{sec:tradeoff}

\safhire{}~\cite[Remark~4]{BCD+25} suggests increasing
the output noise level as a possible defensive
countermeasure; the study below quantifies the utility cost.
Since its shuffle-DP argument in Theorem~\ref{thm:tradeoff} treats each
round's convolution output as a local randomizer,
this proposal is naturally interpreted as increasing
the noise injected at every layer.
Table~\ref{tab:tradeoff} instantiates this sweep
under two defender models.
\emph{Per-layer} injects clipped-Gaussian noise at
the output of every convolutional layer, matching the
defense suggested in Remark~\ref{rem:dp-variants}, whereas \emph{Proxy},
which we retain only as a reference point, injects
noise solely at the final classifier output and
therefore upper-bounds the defender's attainable
utility.%
\footnote{The closest prior noise-based defense
study~\cite{TZC+25} injects Gaussian noise into input
embeddings against input-privacy attacks on
transformers, which differs both in attack target and
injection point. To the best of our knowledge, no
prior work quantifies per-layer noise as a defense
against model-confidentiality attacks in this scheme
class.}

\begin{table}[t]
\centering
\caption{Confidentiality/utility tradeoff on ResNet-20 and ResNet-56
(CIFAR-10, 8-bit). \emph{Atk err}: mean per-column $\ell_\infty$ error.
\emph{FP}: fingerprint top-1 among 20 checkpoints (ResNet-20 only).
\emph{Pred.agr}: prediction agreement. \emph{Acc}: per-layer test accuracy
under clipped-Gaussian noise ($N{=}200$ for ResNet-20, $N{=}300$ for ResNet-56).}
\label{tab:tradeoff}
{\footnotesize
\setlength{\tabcolsep}{3pt}
\begin{tabular}{@{}r|rrrr|rrr@{}}
\toprule
& \multicolumn{4}{c|}{ResNet-20} & \multicolumn{3}{c}{ResNet-56} \\
\cmidrule(lr){2-5}\cmidrule(l){6-8}
Mult. & Atk err & FP & Pred.agr & Acc & Atk err & Pred.agr & Acc \\
\midrule
$1{\times}$      & $0$                   & $20/20$ & $100.0\%$ & $82.2$ & $0$                   & $100.0\%$ & $82.2$ \\
$10{\times}$     & $1.4{\times}10^{-2}$  & $20/20$ & $99.8\%$  & $82.2$ & $1.3{\times}10^{-2}$  & $99.8\%$  & $82.2$ \\
$100{\times}$    & $1.1{\times}10^{-1}$  & $20/20$ & $98.3\%$  & $82.2$ & $1.1{\times}10^{-1}$  & $98.2\%$  & $82.2$ \\
$1{,}000{\times}$& $1.02$                & $1/20$  & $63.7\%$  & $59.0$ & $1.01$                & $51.8\%$  & $48.9$ \\
$5{,}000{\times}$& $6.34$                & $1/20$  & $10.6\%$  & $10.1$ & $6.34$                & $9.6\%$   & $10.0$ \\
\bottomrule
\end{tabular}
}
\vspace{-5mm}
\end{table}

Fingerprinting remains effective up to a $100\times$ noise multiplier with
essentially no utility loss: all 20 ResNet-20 models are still identified,
and per-layer agreement remains $98.3\%$.
It breaks down only at noise levels of at least $1{,}000\times$, where
per-layer accuracy under clipped-Gaussian noise has already dropped to
$59.0\%$ on ResNet-20 and $48.9\%$ on ResNet-56; see
Figure~\ref{fig:phase}.
For ImageNet-scale ResNet-50 on Imagenette with $N=20$, utility collapses
between $30\times$ and $70\times$, about an order of magnitude earlier than
on CIFAR, since noise compounds across 53 layers, as shown by the green
curve in Figure~\ref{fig:phase}.

Theorems~\ref{thm:impossibility}--\ref{thm:tradeoff}
are distribution-agnostic:
across all six bounded noise distributions tested
(zero, Gaussian, triangular, Laplace, uniform, and
worst-case $\eta_k\in\{{+}B,{-}B\}$), no distribution
reverses the ordering; adopting a smoother noise law
can at best delay the utility collapse by a constant
factor.
Full per-distribution accuracy figures are in
Appendix~\ref{app:experiments}.

\begin{figure}[t]
\centering
\includegraphics[width=0.5\linewidth]{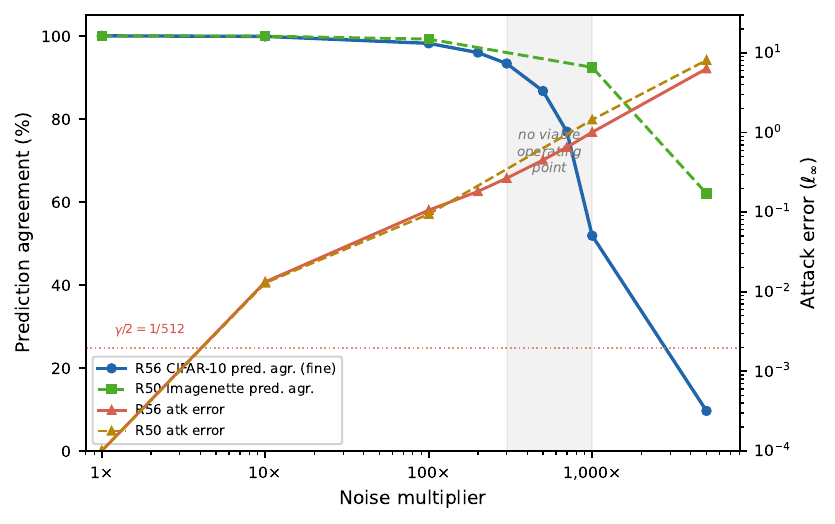}
\caption{Per-layer noise phase transition for ResNet-56 CIFAR-10
(blue, $N{=}100$, nine levels) and ResNet-50 Imagenette (green,
$N{=}20$), clipped-Gaussian at $p{=}256$.
Attack error (red, right axis) and $\gamma/2{=}1/512$ for reference.
Shaded band: attack broken but utility intact; neither curve enters it.}
\label{fig:phase}
\vspace{-3mm}
\end{figure}

In the per-layer noise-injection sweep with i.i.d. noise, no choice of
noise scale, architecture, or bounded noise distribution defeats the attack
while preserving usable accuracy.
Thus, this defense family admits no viable operating point.
Moreover, repeating each column query $k$ times and averaging reduces the
recovery error by a factor of $1/\sqrt{k}$:
at $1{,}000{\times}$ noise, $k=5$ lowers the mean error from $0.92$ to $0.38$.
Thus, increasing the noise raises the query cost quadratically rather than
preventing recovery.
Alternative escape routes, including input validation, zero-knowledge
well-formedness proofs, function secret sharing, and fully server-side FHE,
are discussed in Section~\ref{sec:discussion}.

\section{Discussion}
\label{sec:discussion}

\noindent\textbf{Coverage across the scheme class.}
Theorem~\ref{thm:impossibility} relies on three hypotheses:
(P) the client observes shuffled, noisy outputs of $Wx+b$;
(N) the noise satisfies the correctness bound on a known quantization lattice;
and (C) the client may issue chosen-input queries.
They are worth isolating because closely related schemes satisfy different subsets.
Under its stated threat model, \safhire{}~\cite{BCD+25} satisfies all three at every layer and is therefore the primary target.
STIP~\cite{YZL24} and Centaur~\cite{LCZ+25} violate both (N) and (C): they use exact floating-point computation without quantization-bounded noise, keep intermediate activations encrypted or secret-shared, and reveal only the final output, so Theorem~\ref{thm:impossibility} does not yield a meaningful client-side attack.
Fission~\cite{UMS+25} violates (P), since as an MPC protocol it never reveals a permuted layer output to the client.
PermLLM~\cite{ZCHZ24} falls outside Definition~\ref{def:perm-scheme}, since its permutation acts on the sequence axis rather than the output-channel axis $\pi \in S_m$.

\smallskip\noindent
\textbf{Compound leakage.}
The attack recovers $s_i(W,b)$ rather than the full weight matrix. In \safhire{}, moreover, the final
layer is left unshuffled to enable client-side
$\arg\max$ evaluation. As a result, the recovered
spectra combine with exact logit access to form a
compound leakage channel, exposing both the model's
functionality and its identity, as discussed in
Section~\ref{sec:distillation}.

\smallskip\noindent
\textbf{Necessary conditions for a defense.}
Any viable defense must invalidate at least one hypothesis of
Theorem~\ref{thm:impossibility}.
The known design space offers three broad escape routes.
\emph{Secret linear transforms} mask $Wx+b$ so that, without the key, it is
computationally indistinguishable from random, violating~(P).
\emph{Function secret sharing} across non-colluding servers prevents any one
party from observing the permuted output, violating both~(P) and~(C).
\emph{Fully server-side FHE} withholds intermediate decryptions from the
client, violating~(C).
GAZELLE~\cite{JVC18}, DELPHI~\cite{MLS+20}, and
CrypTFlow2~\cite{RRK+20} instantiate variants of the second route, though at
the bandwidth cost that permutation-based schemes were meant to avoid.

\smallskip\noindent
\textbf{Per-query permutations.}
Because the recovered quantity is a permutation-invariant summary, the guarantee of Theorem~\ref{thm:impossibility} holds for any permutation strategy, including one that draws a fresh permutation for every query. Table~\ref{tab:tfhe_transcript} confirms this empirically: the fresh-per-query regime leaves the maximum recovery error at zero across all 21 layers, so refreshing the permutation on each query does not strengthen the scheme.

\medskip\noindent\textbf{Acknowledgements.}
The authors would like to thank the anonymous reviewers for their helpful comments.
This work was supported by the Institute of Information \& Communications Technology Planning \& Evaluation~(IITP) grant funded by the Korea government(MSIT) (No. RS-2024-00399491, 75\%). Jiseung Kim was also supported by the National Research Foundation of Korea (NRF) grant funded by the Korean government (MSIT) (No. RS-2026-25477662, 25\%).

\bibliographystyle{splncs04}
\bibliography{ref}

\appendix

\section{CIFAR-10 / ResNet-20 Experimental Details}
\label{app:cifar}
Training details for the 20 ResNet-20 checkpoints used in
Section~\ref{sec:attack}: SGD with learning rate $0.1$, momentum $0.9$,
weight decay $10^{-4}$, MultiStepLR milestones at epochs 5 and 8,
batch size 128, 10 epochs, and seeds $0,\ldots,19$, yielding mean test
accuracy 82.8\%.
Theorem~\ref{thm:impossibility} is deterministic under the correctness
bound, so zero-error recovery does not depend on teacher accuracy;
longer training affects only the fingerprint-separation constant
$\Delta_\mathcal{F}^\phi$ in Section~\ref{sec:attribution}.
The CIFAR-100 comparison there, trained for 200 epochs, gives the larger
value $\Delta_\mathcal{F}^\phi = 0.565$ and still separates all 190 pairs.
All 21 convolution modules use \texttt{bias=False}, so before batch-normalization
folding we have $b=0$ and the raw convolution reduces to the exact weight multiset.
In the artifact code, BN layers are folded into the preceding convolution via
the standard deployment transform
$W \mapsto \gamma W/\sigma_{\mathrm{BN}}$ and
$b \mapsto \gamma(\mu_{\mathrm{conv}}-\mu_{\mathrm{BN}})/\sigma_{\mathrm{BN}}+\beta$.
Accordingly, the recovered object is the folded affine pair $(W',b')$ that the
client actually sees as the layer output.
Residual branches $y=F(x)+x$ are probed branchwise, consistent with
\safhire{}'s per-layer ciphertext reply pattern.
The per-layer input dimensions $d_l$ range from 16 for
\texttt{layer2.0.shortcut} to 576 for \texttt{layer3.2.conv2}, yielding the
per-model budget
$\sum_{l=1}^{21}(d_l+1)=5{,}712$
from Theorem~\ref{thm:impossibility}.

\section{Tightness of the Correctness Bound}
\label{app:proofs}

\begin{proposition}
\label{prop:tight}
If $\Delta > \gamma/2$, there exists $\eta$ with
$\|\eta\|_\infty < \Delta$ under which no deterministic
lattice-valued algorithm recovers $\sort(W_{:,i}+b)$
correctly.
\end{proposition}

\begin{proof}
Fix $\eta^* \in (\gamma/2,\min(\gamma,\Delta))$, which is nonempty since
$\Delta>\gamma/2$.
Two instances share the same observation
$o := k\gamma+\eta^*$:
instance~A with true value $k\gamma$ and noise
$\eta^*$, and instance~B with true value $(k{+}1)\gamma$
and noise $\eta^*-\gamma$, where
$|\eta^*-\gamma|=\gamma-\eta^*<\gamma/2<\Delta$ since $\eta^*<\gamma$.
Distinct true values and identical observation force any deterministic
algorithm to err on at least one instance.
Since $\sort$ is deterministic, the ambiguity carries over to the
sorted spectrum.  \qed
\end{proof}

\section{Additional Experimental Validation}
\label{app:experiments}

\smallskip
\noindent\textbf{Noise distribution robustness.}
We test six bounded noise distributions with $\|\eta\|_\infty < \Delta$
against the per-layer noise-injection defense of Section~\ref{sec:tradeoff}.
Under zero noise, ResNet-20 and ResNet-56 retain $82.3\%$ and $82.2\%$ accuracy at both
$m=100$ and $m=1{,}000$.
Under Gaussian noise, both remain at $82.2\%$ for $m=100$, then drop to
$59.0\%$ for ResNet-20 and $48.9\%$ for ResNet-56 at $m=1{,}000$.
Heavier-tailed distributions collapse earlier: at $m=1{,}000$, triangular
reaches $41.3\%$ / $29.8\%$, Laplace $21.6\%$ / $13.6\%$, and uniform
$18.7\%$ / $12.5\%$.
The worst-case deterministic envelope $\eta_k\in\{{+}B,{-}B\}$ is most
damaging, reaching $76.8\%$ / $75.6\%$ at $m=100$ and $10.0\%$ for both
architectures at $m=1{,}000$.
No distribution reverses this ordering, supporting the distribution-agnostic
guarantee of Theorem~\ref{thm:impossibility}.

\smallskip
\noindent\textbf{ImageNet-scale experiments.}
Table~\ref{tab:imagenet_recovery} uses torchvision
\texttt{IMAGENET1K\_V1} weights (eight CNNs and
ViT-B/16). The ResNet-50 per-layer curve
(Fig.~\ref{fig:phase}) uses a fixed $500$-image
Imagenette-val subset ($50$/class,
\texttt{random.seed(0)}) with top-1 against real
ImageNet-1K labels via the WNID$\to$index map. The
lineage variant $V_1^{\mathrm{ft}}$ is V1 after two
epochs of SGD on the same subset (LR~$10^{-3}$,
batch $32$).

\smallskip
\noindent\textbf{Bias-having MLP.}
To validate the bias case of Theorem~\ref{thm:impossibility},
we train a two-layer MLP with biases on MNIST
($97.7\%$ test accuracy; $\mathtt{fc1}: 128 \times 784$,
$\mathtt{fc2}: 64 \times 128$, $8$-bit precision).
For every column, $\sort(W_{:,i}{+}b) \ne
\sort(W_{:,i}) + \sort(b)$, confirming that the bias
cannot be separated from the recovered multiset;
round-then-sort recovers $\sort(W_{:,i}{+}b)$ exactly
on all columns of both layers.

\section{Distillation Ablation}
\label{app:distillation}

We compare four knowledge-distillation variants that differ only in
the spectral prior drawn from the recovered $L(W,b)$:
\emph{Baseline} (logits only);
\emph{First} ($\sort(W_{:,i})$ from \texttt{conv1});
\emph{Col-norms} (column norms across all 21 convolutional layers);
\emph{Full} (column norms + per-layer weight-value remapping).
A \emph{Scrambled} control replaces teacher column norms with those
of an independently trained model to isolate structural vs.\
teacher-specific gains.

\begin{table}[t]
\centering
\caption{Logit distillation test accuracy (\%) on CIFAR-10 under
four spectrum-prior ablations (mean~$\pm$~std, 10 seeds).
Teacher accuracies: $82.92\%$ (ResNet-20) and $82.47\%$ (ResNet-56).}
\label{tab:kd_ablation}
\setlength{\tabcolsep}{3pt}
\footnotesize
\begin{tabular}{@{}llcccc@{}}
\toprule
Architecture & $q$ & Baseline & First & Col-norms & Full \\
\midrule
\multirow{3}{*}{ResNet-20}
 & 1{,}000 & $47.21\pm1.21$ & $48.60\pm1.51$ & $50.09\pm0.94$ & $51.11\pm1.06$ \\
 & 2{,}000 & $56.16\pm1.60$ & $58.84\pm1.20$ & $59.87\pm0.91$ & $60.07\pm0.82$ \\
 & 5{,}000 & $71.58\pm1.67$ & $72.71\pm2.14$ & $74.66\pm0.50$ & $74.12\pm0.97$ \\
\midrule
\multirow{2}{*}{ResNet-56}
 & 2{,}000 & $53.28\pm1.68$ & $53.60\pm1.10$ & $55.38\pm1.37$ & $54.53\pm1.64$ \\
 & 5{,}000 & $70.99\pm1.45$ & $72.46\pm2.07$ & $73.50\pm0.92$ & $72.53\pm1.10$ \\
\bottomrule
\end{tabular}
\vspace{-5mm}
\end{table}

On ResNet-20, \emph{Full} is strictly best at $q\le 2{,}000$;
\emph{Col-norms} becomes competitive at $q=5{,}000$
($74.66\pm0.50$ vs.\ $74.12\pm0.97$).
On ResNet-56, \emph{Col-norms} already matches or exceeds \emph{Full}
at both budgets. The \emph{Scrambled} control ($74.02\pm1.10$ at
$q=5{,}000$) is within one std of \emph{Col-norms}, indicating the
gain is largely structural.

\smallskip\noindent\textbf{Configuration.}
Teachers: \texttt{resnet20\_seed0.pt} ($82.92\%$) and
\texttt{resnet56\_seed0.pt} ($82.47\%$); students trained from
scratch with SGD (LR~$0.05$, momentum~$0.9$, wd~$10^{-4}$,
cosine anneal, $100$ epochs, batch~$128$).
KL on softmax logits at $T{=}4$; column-norm regularizer
weight $\lambda{=}5$ for \emph{Col-norms}/\emph{Full}.
Query budgets drawn from a fixed $5{,}000$-image random subset
(\texttt{random.seed(0)}); $10{,}000$-image test set; $10$ seeds.

\end{document}